\documentclass{eptcs}

\usepackage[pdftex]{graphicx}
\usepackage{amsmath,amsfonts,latexsym,amssymb,epsfig,verbatim}
\usepackage[mathcal]{euscript}

\newcommand{\VC}[1]{}

\newcommand{\couic}[1]{}

\newcommand{\ports}{\pi}

\newtheorem{theorem}{Theorem}[section]
\newtheorem{lemma}[theorem]{Lemma}

\newenvironment{proof}[1][Proof]{\begin{trivlist}
\item[\hskip \labelsep {\bfseries #1}]}{\end{trivlist}}
\newenvironment{definition}[1][Definition]{\begin{trivlist}
\item[\hskip \labelsep {\bfseries #1}]}{\end{trivlist}}

\newcommand{\qed}{\nobreak \ifvmode \relax \else
      \ifdim\lastskip<1.5em \hskip-\lastskip
      \hskip1.5em plus0em minus0.5em \fi \nobreak
      \vrule height0.75em width0.5em depth0.25em\fi}

\title{Intrinsic Universality of Causal Graph Dynamics}
\author{Simon Martiel
\institute{Universit\'e Nice Sophia Antipolis, I3S, CNRS UMR 7271, 06903 Sophia Antipolis, BP121, France}
\thanks{This work was supported by the John Templeton Foundation, grant ID 15619}
\thanks{This work was supported by the French National Research Agency, project EMC (ANR\_09\_BLAN\_0164)}
\email{martiel@i3s.unice.fr}
\and Bruno Martin
\institute{Universit\'e Nice Sophia Antipolis, I3S, CNRS UMR 7271, 06903 Sophia Antipolis, BP121, France}
\thanks{This work was supported by the French National Research Agency, project EMC (ANR\_09\_BLAN\_0164)}
\email{bruno.martin@unice.fr}
}

\def\titlerunning{CGD Intrinsic Universality}
\def\authorrunning{S. Martiel \& B. Martin}

\begin{document}
\maketitle

\begin{abstract}
Causal graph dynamics are transformations over graphs that capture two important symmetries of physics, namely causality and homogeneity. They can be equivalently defined as continuous and translation invariant transformations or functions induced by a local rule applied simultaneously on every vertex of the graph. Intrinsic universality is the ability of an instance of a model to simulate every other instance of the model while preserving the structure of the computation at every step of the simulation. In this work we present the construction of a family of intrinsically universal instances of causal graphs dynamics, each instance being able to simulate a subset of instances.\\


\end{abstract}

\section{Introduction}

\noindent \emph{Causal Graph Dynamics} have been introduced in \cite{ArrighiCGD,ArrighiCayley} as a generalization of cellular automata on arbitrary graphs. 
Basically, causal graph dynamics are graphs transformations that capture the notion of causality (bounded speed of information) and homogeneity (the rules are translation invariant). It has been proved in \cite{ArrighiCGD} that this set of transformations over graphs corresponds exactly to functions induced by a local rule applied simultaneously on every vertex of a graph, the local rule being able to change the labels on the vertices and also the topology of the graph. Like cellular automata, they can be defined in two different ways. They can be seen as continuous and translation invariant functions over graphs for a Gromov-Hausdorff metric or as functions induced by a local rule applied synchronously on every vertex of a graph. While the first definition is very axiomatic and very close to the physical notions of causality and homogeneity, the second definition provides a constructive characterization of the model. As these two definitions of continuous and translation invariant transformations and local rule induced transformations are equivalent, we will focus here on the most constructive one.

\noindent \emph{Intrinsic universality} is the property of having one instance of the model of computation able to simulate all other instances while preserving the structure of the computation. This notion of preserving the structure of the computation has a precise meaning when studying models where a notion of space exists and has already been intensively studied for the cases of cellular automata \cite{OllingerCSP,Durand-LoseIntrinsic1D,cracow} or quantum cellular automata~\cite{ArrighiFI,ArrighiPQCA}. 

Intrinsic universality is a more constrained property than universality. Indeed we could just exhibit some local rule simulating a universal Turing machine or a universal cellular automaton (as in \cite{ArrighiCayley}) and conclude that this local rule is universal as it can simulate any instance of our model. In the case of causal graph dynamics, the computation is distributed on every vertex of a graph. The formalism provided in the previous work on the model seems powerful enough to define a notion of intrinsic simulation and thus the question of intrinsic universality arises naturally. In this work we present the construction of a family of intrinsically universal instances of causal graphs dynamics, each instance being able to simulate a subset of instances.
The paper is organized as follows. In section 2, graphs and causal graphs dynamics are formally introduced together with a formal definition of universality and intrinsic simulation. In section~3, two methods are described to encode the initial graph and the local rule to be simulated. In section 4, a universal construction machine is presented together with a family of intrinsically universal instances of the model. In a last section, conclusion and further works are given.

\section{Definitions and notations}

Although causal graph dynamics can be defined in a very abstract and axiomatic way (see \cite{ArrighiFI}), we provide here a constructive description. This constructive definition is given by the formalism of functions over generalized Cayley graph induced by a local rule (introduced in \cite{ArrighiCayley}).\\

\noindent Notations:
\begin{itemize}
\item $V$ is an infinite, (possibly uncountable) set, containing all possible vertex names.
\item $\pi$ is always a finite set of the form $\{1,...,d\}$ and is called the set of ports. 
\item $\Pi=\pi\times\pi$ is an alphabet and we denote by ``$.$'' the operation of concatenation of words over this alphabet.
\item $u:i$ stands for port $i$ of vertex $u$.
\item $\Sigma$ denotes a finite set of labels.
\item $S$ denotes a finite set of the form $\{\varepsilon , 1, 2,...,s\}$ where $\varepsilon$ denotes the empty word.
\end{itemize}

\subsection{Generalized Cayley graphs and localizable functions}

The first definition introduces the notion of \emph{labeled graphs}. Our graphs have a bounded degree. Every vertex has a set of ports on which edges are connected. Every port can receive at most one edge, edges are undirected. Moreover, every graph is assumed to be connected. 
\begin{definition}[Graph]\label{def:graphs}
A labeled {\em graph} $G$ is given by 
\begin{itemize}
\item[$\bullet$] A (at most countable) subset $V(G)$ of $V$, whose elements are called {\em vertices}.
\item[$\bullet$] A finite set $\pi=\{1,...,d\}$, whose elements are called {\em ports}.
\item[$\bullet$] A set $E(G)$ of two element subsets of $V(G):\pi$, such that for all vertex $u\in V(G)$ and port $i\in \pi$, $u:i$ appears at most once in $E(G)$. Elements of $E(G)$ are called edges.
\item[$\bullet$] A function $\sigma_G: V(G) \rightarrow \Sigma$ associating to each vertex $v$ some label $\sigma_G(v)$.
\end{itemize}
The set of labeled graphs with set of ports $\pi$ and with labels in $\Sigma$ is denoted $\mathcal{G}_{\pi,\Sigma}$.
\end{definition}

The definition of generalized Cayley graphs was introduced in \cite{ArrighiCayley}. They are graphs described relatively to a pointed vertex and can be seen as Cayley graphs where the internal operation associated to each generator is not defined on every vertex .

First we present Generalized Cayley graphs informally. Generalized Cayley graphs are a way to describe a pointed graph by associating to each vertex the set of all possible paths starting from the pointed vertex and leading to the vertex:
\begin{itemize}
\item[$\bullet$] $L$ is the set of all paths of the graph. Paths are of the form  $u.ab$ where $a,b$ are ports and $u$ is a path. The empty path is denoted $\varepsilon$. $u.ab$ means that after following the path $u$, we can exit the current vertex using port $a$ and enter next vertex using port $b$ and thus extend the path. Figure \ref{fig:paths} provides an illustration of the notion of path.
\begin{figure}
\begin{center}
\includegraphics[scale=1.2]{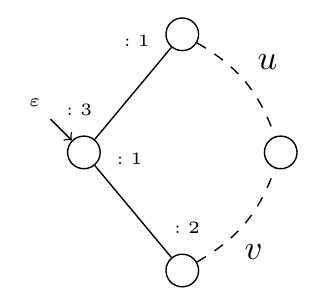}
\end{center}
\caption{In this generalized Cayley graph, the rightmost vertex will be in the equivalence class of the paths $31.u$ and $12.v$ where $u$ and $v$ are paths. The equivalence class of $\varepsilon$ also contains the paths $31.13$, $(31.13)^2$, \ldots \label{fig:paths} }
\end{figure}

\item[$\bullet$] Vertices are the equivalence classes of $\equiv_L$. Two paths are equivalent if they lead to the same vertex.
\item[$\bullet$] The graph is naturally pointed by the equivalence class of the empty path $\varepsilon$ (which necessarily belongs to the graph due to the first property). Thus we can talk about the pointed vertex of a generalized Cayley graph.

\end{itemize}
We can give a more intuitive meaning for the required properties:
\begin{itemize}
\item[$\bullet$] If the composed path $u.v$ is in the graph, then the path $u$ is in the graph.
\item[$\bullet$] If $u$ and $u'$ lead to the same vertex and $u.v$ is in the graph, then $u'.v$ is in the graph and leads to the same vertex as $u.v$.
\item[$\bullet$] If $u.(ab)$ leads to some vertex, then it possible to go back through the last edge to reach the vertex pointed by $u$ with $u.(ab).(ba)$ (every path has an ``inverse'' path, hence the analogy with the generators of a Cayley graph).
\item[$\bullet$] A vertex cannot be connected to two different vertices using the same port.
\end{itemize}

Now we can give the formal definition.

\begin{definition}[Generalized Cayley graph]
Let $L\subseteq \Pi^*$ be a language and $\equiv_L$ be an equivalence relation on this language. The tuple $(L,\equiv_L)$ is a generalized Cayley graph if it satisfies:
\begin{enumerate}
\item $\forall u,v\in \Pi^*\quad u.v \in L \Rightarrow u\in L$
\item $\forall u,u'\in L,\forall v\in\Pi^*\quad (u\equiv_L u'\,\wedge\, u.v\in L) \Rightarrow (u'.v\in L \,\wedge\, u'.v\equiv_L u.v)$
\item $\forall u\in L,\forall a,b\in \ports \quad u.ab\in L \Rightarrow (u.ab.ba\in L \,\wedge\, u.ab.ba\equiv_L u)$
\item $\forall u,u'\in L,\forall a,b,c \in \ports \quad (u\equiv_L u'\wedge u.ab\in L \wedge u'.ac\in L) \Rightarrow b=c.$
\end{enumerate}
This definition naturally extends to labeled generalized Cayley graphs. The set of generalized Cayley graphs with set of ports $\pi$ and labels $\Sigma$ is written $\mathcal{X}_{\pi,\Sigma}$.
\end{definition}
In the following section, we denote by $\sim G$ the generalized Cayley graph built from a graph $G$ with paths starting from the vertex with name $\varepsilon$. The operation $\sim$ is only defined on graph having a vertex named $\varepsilon$.

Notice that when considering a generalized Cayley graph, we will commonly assimilate $\tilde{u}$ (the equivalence class of $u$) and $u$ and thus we might talk about the vertex $u\in X$ for $X\in\mathcal{X}_{\pi,\Sigma}$ even though $u$ is a path.

In addition to the above properties, some operations are available for generalized Cayley graphs:
\begin{itemize}
\item[$\bullet$] The shift operation consists in moving the pointer of the graph along a given path. The graph $X$ shifted along the path $u\in L$ is denoted $X_u$.
\item[$\bullet$] The (rank $r$) neighboring operation consists in preserving only the disk of radius $r$ centered on the pointer. As generalized Cayley graphs are pointed we can define the disk of radius $r$ of a graph. The disk of radius $r$ of $X$ is denoted $X^r$. The set of disks of radius $r$, of ports $\pi$ and of labels in $\Sigma$ is written $\mathcal{X}^r_{\pi,\Sigma}$.
\end{itemize}

Moreover, we need a prefixing operation acting on graphs of the set $\mathcal{G}_{\pi,\Sigma}$. In the following definitions, $u.G$ with $u\in \Pi^*$ and $G$ a graph, stands for the graph $G$ where names of vertices are prefixed with $u$.

The next definition is a requirement for defining the union of two graphs after applying the local rule. Notice that it is a definition on graphs in $\mathcal{G}_{\pi,\Sigma}$.

\begin{definition} [Consistency]
Consider two labeled graphs $G$ and $H$ in ${\cal G}_{\pi,\Sigma}$, they are {\em consistent} if and only if for all vertices $u,v,w$ and ports $k,l,p$ we have:
$$ \{u:k,v:l\} \in E(G) \wedge  \{u:k,w:p\}\in E(H) \Rightarrow v=w \wedge l=p$$

\noindent and for all vertices $u\in V(G)\cap V(H)$ we have that $\sigma_G(u)=\sigma_H(u)$.
Two graphs are trivially consistent if  $V(G)\cap V(H)$ is empty.
\end{definition}

Less formally: Two graphs are consistent if they agree on their intersection. They must agree on the label of the vertices and on the used ports. \medskip

We next define the notion of \emph{local rule}.
In a graph generated by a local rule, names of vertices have a particular meaning. When applied on a disk $X^r_u\in\mathcal{X}^r_{\pi,\Sigma}$, the local rule $f$ produces a graph $f(X^r_u)$ such that the names of its vertices are sets of elements of the form $u.z$ with $u$ a path of $X^r_u$ and $z$ a suffix in $S$.
The conventions taken are such that integer $z$ stands for the `successor number $z$'. Hence the vertices designated by $\varepsilon,1,2\ldots$ are successors of the vertex $\varepsilon$, whereas those designated by $u.1,u.2\ldots$ are successors of its neighbor $u\in X^r$. For instance a vertex named $\{1,ab.2\}$ is understood to be both the first successor of vertex $\varepsilon$ and the second successor of the vertex attained by the path $ab$. Such a vertex can be designated by $1, ab.2$ or $\{1,ab.2\}$.
\begin{definition}[Local rule]
A (possibly partial) function $f$ from ${\cal X}^r_{\pi,\Sigma}$ to ${\cal G}_{\pi,\Sigma}$ is a {\em local rule} if and only if :
\begin{itemize}
\item[$\bullet$] For all $X$, the vertices of $f(X)$ are disjoint subsets of $V(X).S$ and $\varepsilon\in f(X)$,
\item[$\bullet$] There exists a bound $b$ such that for all disks $X^{r+1}$, $|V(f(X^{r+1}))|\leq b$,
\item[$\bullet$] For any disk $X^{r+1}$ and any $u\in X^{0}$ we have that $f(X^r)$ and $u.f(X_u^r)$ are non-trivially consistent,
\item[$\bullet$] For any disk $X^{3r+2}$ and any $u\in X^{2r+1}$ we have that $f(X^r)$ and $u.f(X_u^r)$ are consistent.
\end{itemize}
\end{definition}
 The conditions of consistency are here to ensure that if the local rule is applied on two ``close'' vertices of the same graph, the two resulting graphs will be intersecting and consistent.

A local rule is a mathematical object which can be characterized by four parameters:
\begin{itemize}
\item $|\pi|$ the degree of the graphs it is applied on,
\item $\Sigma$ the set of vertex labels,
\item $r$ the  radius of the disks it is applied on,
\item $b$ the maximal size of its images.
\end{itemize}
The set of local rules of parameters $(|\pi|,\Sigma,r,b)$ is denoted $\mathcal{F}_{\pi,\Sigma,r,b}$.

\begin{definition}[Union]
The union $G\cup H$ of two consistent graphs $G$ and $H$ is defined as follows:
\begin{itemize}
\item $V(G\cup H)=V(G)\cup V(H)$,
\item $E(G\cup H)=E(G)\cup E(H)$.
\end{itemize}

\end{definition}

The definition of localizable function describes how these local rules can be used to induce a global function that acts on graphs of arbitrary size.

\begin{definition}[Localizable function]\label{def:localizable}
A function $F$ from ${\cal X}_{\Sigma, \ports}$ to ${\cal X}_{\Sigma, \ports}$ is said to be {\em localizable} if and only if there exists a radius $r$ and a local rule $f$ from ${\cal X}^r_{\Sigma, \ports}$ to ${\cal G}_{\Sigma, \ports}$ such that for all $X$, $F(X)$ is given by the equivalence class, with $\varepsilon$ taken as the pointer vertex, of the graph
$$\sim \bigcup_{u\in X} u.f(X_u^r).$$
where $\sim G$ constructs the generalized Cayley graph having the same structure as $G$, starting from the vertex with name $\varepsilon$.
\end{definition}
\medskip
We provide below an example of a local rule which splits a vertex into 4 vertices.

\noindent {\em Inflating grid.} This example is taken from \cite{ArrighiCGD}. 
Each vertex gives birth to four distinct vertices, such that the structure of the initial graph is preserved, but inflated. The graph has maximal degree $4$, as the ports take their names in $\pi=\{n,s,e,w\}$. Vertices and edges are unlabelled. The local rule is of radius zero.
The standard case of the local rule is described in Figure \ref{fig:inflatinggeneral}. It generates a subgraph of $12$ vertices, with names serving as identification information, so that the generated graphs glue back together. 
\begin{figure}[htpb]
\centering
\includegraphics[scale=0.4, clip=true, trim=0cm 19cm 0cm 0cm]{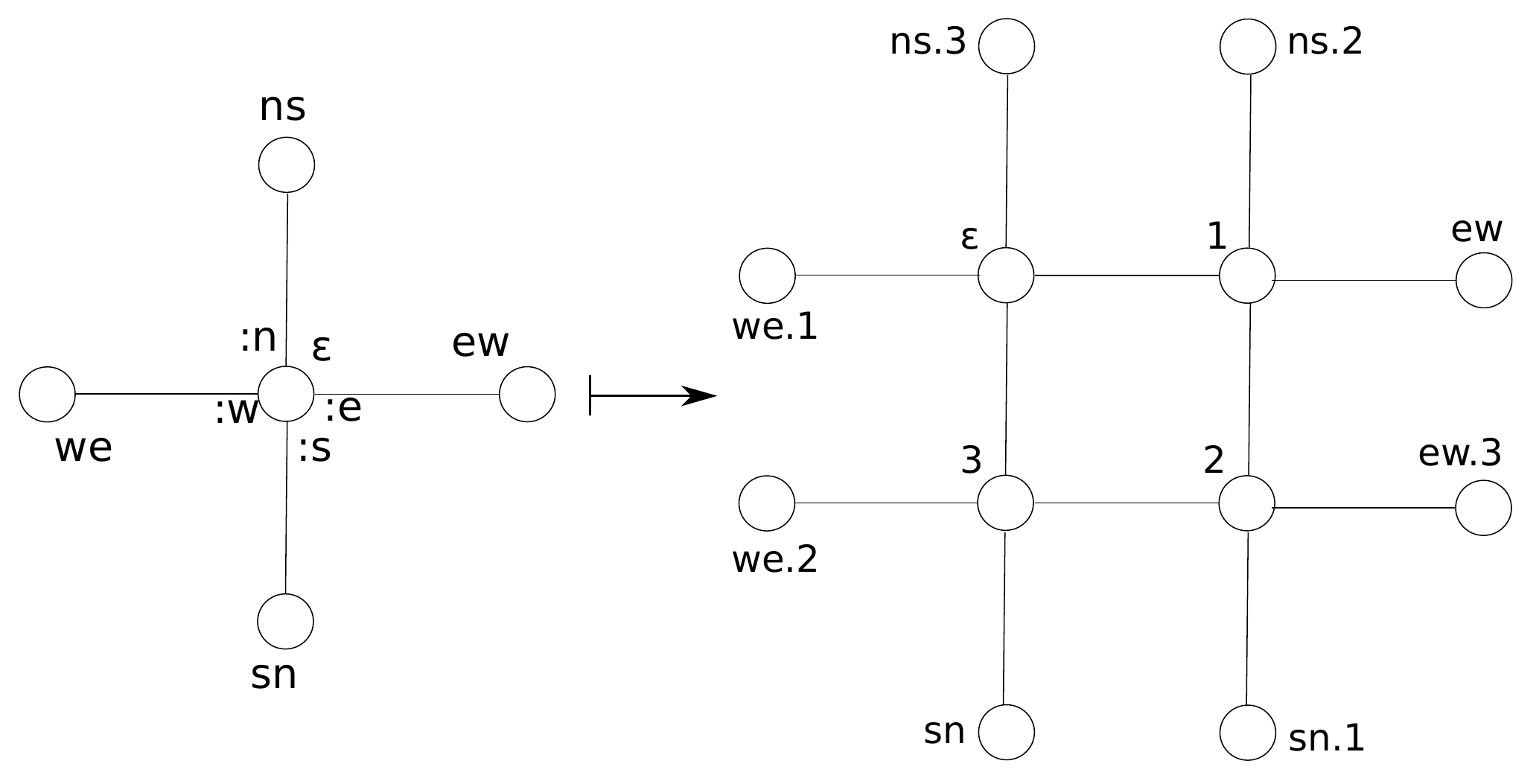}
\caption{Standard case of the local rule for the inflating grid (four neighbours). \label{fig:inflatinggeneral}}
\end{figure}
For a complete definition we would also have to include the boundary cases, when vertex $\varepsilon$ is surrounded by less than $4$ neighbours (see Figure \ref{fig:inflatingborder}, for instance). 
\begin{figure}[htpb]
\centering
\includegraphics[scale=0.45, clip=true, trim=0cm 21cm 0cm 0cm]{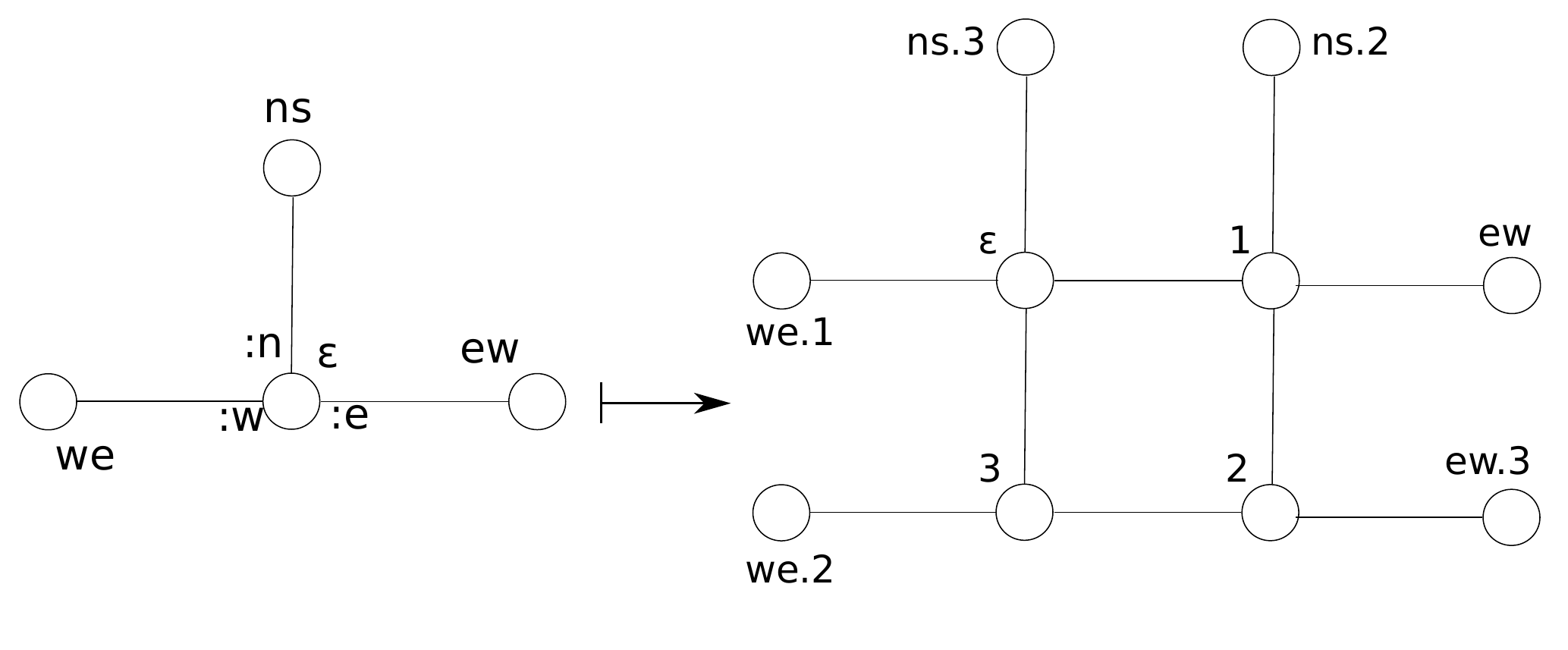}
\caption{Boundary case of the local rule for the inflating grid (three neighbours). \label{fig:inflatingborder}}
\end{figure}
Moreover, we would also have to include the cases when ports have not been set up in the natural manner (e.g. if there is an $ee$ connection), or when the vertex has a loop, etc. None of these is a problem.

The global dynamics is obtained by gluing the different $u.f(X^r_u)$, identifying the vertices having the same set of names, as illustrated in Figure \ref{fig:inflatingglue}.
\begin{figure}[htpb]
\centering
\includegraphics[scale=0.45, clip=true, trim=0cm 4cm 0cm 0cm]{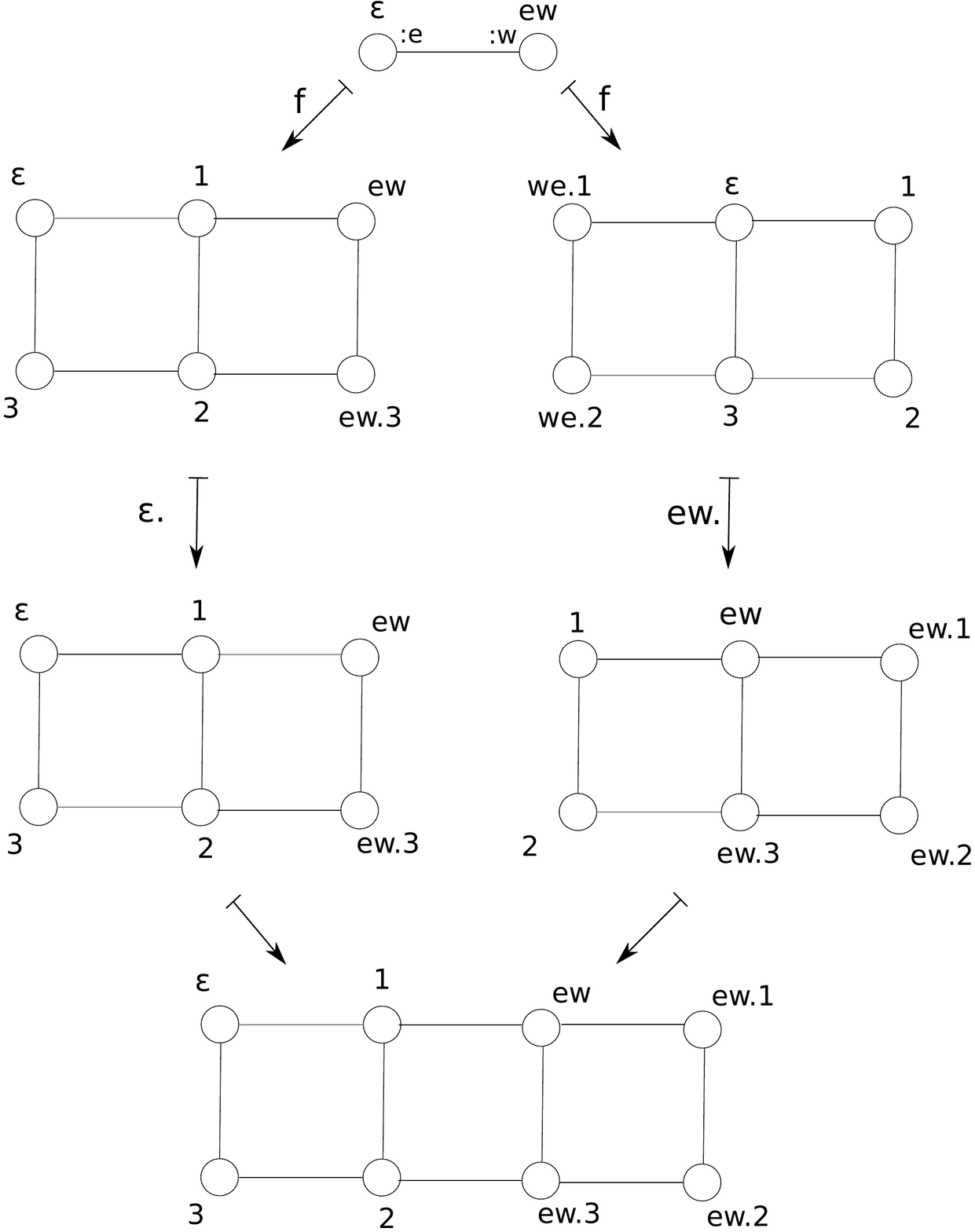}
\caption{The dynamics applied upon a pair of vertices.\label{fig:inflatingglue} {\small First the local rule is applied at each vertex. Then its images are shifted by the names of the vertices. Finally, they are glued back together.}}
\end{figure}

\subsection{Uniform Continuity and Shift invariance}
Intrinsic simulation for cellular automata is often based on the notion of grouping, i.e to each cell of the simulated automaton corresponds a group of cells in the simulating automata (see \cite{theyssierIS} for instance). Due to the variety of the graphs we are working on (they are not necessarly grids), we need a generalized notion of cell grouping (or here, vertex grouping). For this we need two notions: Uniform Continuity and shift invariance.

\medskip

\noindent {\bf Uniform Continuity}
We introduce the following metric over $\mathcal{X}_{\pi,\Sigma}$:
\begin{definition}[Gromov-Hausdorff-Cantor metrics]\label{def:metric}
~\\Consider the function
\begin{align*}
d:{\cal X}_{\Sigma, \ports}\times{\cal X}_{\Sigma, \ports} &\longrightarrow {\mathbb R}^+\\
(X,X')&\mapsto d(X,Y)=0\quad \textrm{if }X=Y\\
(X,X')&\mapsto d(X,Y)=1/2^r\quad \textrm{otherwise}
\end{align*}
where $r$ is the minimal radius such that $X^r \neq Y^r$.\\
The function $d(.,.)$ is such that for $\epsilon>0$ we have (with $r=\lfloor - \log_2(\epsilon)\rfloor$):
$$d(X,Y)<\epsilon \Leftrightarrow X^r = Y^r.$$
\end{definition}

\noindent Once endowed with such a metric, the following result arises:
\begin{lemma}[Compactness] \label{lem:compactness}
$({\cal X}_{\Sigma, \ports},d)$ is a compact metric space, i.e. every sequence admits a converging subsequence.
\end{lemma}

Morover, we now have a notion of uniform continuity for functions over generalized cayley graphs:

\begin{definition}[Uniformly continuous function]\label{def:continuity}
$F:\mathcal{X}_{\pi,\Sigma}\rightarrow\mathcal{X}_{\pi',\Sigma'}$ is continuous if and only if:
$$ \forall r,\exists k,\forall X,Y\in \mathcal{X}_{\pi,\Sigma},X^k=Y^k \Rightarrow F(X)^r=F(Y)^r$$
\end{definition}
Notice that, using the compactness of ${\cal X}_{\Sigma, \ports}$, we have that continuous functions over graphs and uniformly continuous functions.

\noindent {\bf Shift invariance}
We will not formally define shift invariance here; it has been properly introduced in \cite{ArrighiCayley}. Though easy to intuitively understand, shift invariance requires some additional properties in the definition of functions over graphs and thus considerably complexifies the notations. We will only give an informal idea of what a shift invariant function over generalized Cayley graphs is.

First of all, generalized Cayley graphs are pointed graphs, as a particular vertex $\varepsilon$ exists. Shifting a graph corresponds to moving the pointer on another vertex of the graph. That being told, a shift invariant function can be seen as a function which transforms a graph in another graph whose structure does not depend on the position of the pointer. More precisely, any shift in the initial graph will only induce a shift of the image graph.

\subsection{Universal machine}
Turing's construction of a universal machine is not the only way to define universality. While his formalism captures the idea we have in mind of what would be nowadays called an interpreter, another less formal definition suggested by von Neumann, namely the universal construction machine, is closer to the notion of a compiler. Von Neumann's idea, directly inspired by Turing's universal machine, is that there must exist a machine (not necessarily a Turing machine) which, when provided a suitable description of an instance of a computational model, constructs a copy of it. This definition is particularly useful when considering the problem of self-reproduction~\cite{arbib}. The most classical example to illustrate this definition is the uniform generation of Boolean networks where a Turing machine receives as an input the standard encoding of the circuit and its size and explicitely generates the corresponding Boolean network~\cite{badiga88}. Though detached from any mathematical formalism, this notion particularly fits to our model, in as much as we are allowed to modify the topology and have enough freedom to design such a machine inside the model iteself. In our case we want to build this universal construction machine inside our model and such that given the encoding of an initial graph $X$ and the encoding of a local rule $f$, the machine constructs a simulation of $f$ applied on $X$. To do so, we first define the notion of intrinsic simulation. Just before, we introduced the notion of uniformly continuous and shift invariant functions over generalized Cayley graphs which we will use here to define intrinsic simulation.

\begin{definition}[Intrinsic simulation]
Let $X_1\in \mathcal{X}_{\pi_1,\Sigma_1}$ and $X_2\in  \mathcal{X}_{\pi_2,\Sigma_2}$ be two graphs and let $f_1\in\mathcal{F}_{\pi_1,\Sigma_1,r_1,b_1}$ and $f_2\in\mathcal{F}_{\pi_2,\Sigma_2,r_2,b_2}$ be two local rules inducing respectively $F_1$ and $F_2$. We say that $(f_1,X_1)$ intrinsically simulates $(f_2,X_2)$ with constant slow down $\delta$ if and only if there exists a computable, uniformly continuous and shift invariant injection $E:\mathcal{X}_{\pi_2,\Sigma_2}\rightarrow \mathcal{X}_{\pi_1,\Sigma_1}$ and an integer $\delta \in \mathbb{N}^*$ such that:
$$ \forall n\in \mathbb{N}, F_1^{\delta n}(X_1)=E(F_2^n(X_2)) $$
\end{definition}
Now we have a notion of simulation, we can define our universal machines and rules:

\begin{definition}[Universal machine and universal rule]
A universal construction machine is a 5-tuple\\$(X_M,f_M,enc_X,enc_f,f_{\operatorname{univ}})$ such that:
\begin{itemize}
\item $X_M\in\mathcal{X}_{\pi',\Sigma'}$ is a graph with at least $2$ ports available,
\item $f_M$ is a local rule on disks of $\mathcal{X}^{r_M}_{\pi',\Sigma'}$,
\item $enc_X: \mathcal{X}_{\pi,\Sigma}\rightarrow \mathcal{X}_{\pi',\Sigma'}$ an injective computable function encoding the initial graph, 
\item $enc_f: \mathcal{F}_{\pi,\Sigma,r,b}\rightarrow \mathcal{X}_{\pi',\Sigma'}$ an injective computable function encoding the local rule to simulate,
\item $f_{\operatorname{univ}}$ is a local rule on disks of $\mathcal{X}^r_{\pi',\Sigma'}$,
\end{itemize}
and such that for all graph $X$ and local rule $f$, successive applications of the localizable function induced by $f_M$ on the graph composed of $X_M$ connected to $enc_X(X)$ and $enc_f(f)$ will constructs a graph  $X_f$ such that $(f_{\operatorname{univ}},X_f)$ intrinsically simulates $(f,X)$. The local rule $f_{\operatorname{univ}}$ is said to be universal.
\end{definition}
Less formally, the machine itself is a graph $X_M$ to which we will connect the encoding of the initial graph $X$ and the encoding of $f$. After a certain number of iterations of the local rule $f_M$ (possibly infinitely many iterations, if $X$ is infinite), the initial state $X_f$ of the simulation is ready. The successive iteration of the simulation can be computed using a universal function $f_{\operatorname{univ}}$ (that does not depend on $X$ nor $f$).

The next two sections are devoted to the construction of a universal construction machine and its associated universal local rule for fixed parameters $\pi$,$\Sigma$,$r$ and $b$, which provide a proof to the result:

\begin{theorem}[Intrinsic universality]
For every set of parameters $\pi$,$\Sigma$, $r$ and $b$, there exists a universal machine $(X_M,f_M,enc_X,enc_f,f_{\operatorname{univ}})$ such that for all $X\in\mathcal{X}_{\pi,\Sigma}$ and for all $f\in\mathcal{F}_{\pi,\Sigma,r,b}$, successive applications of the localizable function induced by $f_M$ on the graph composed of $X_M$ connected to $enc_X(X)$ and $enc_f(f)$ will construct a graph  $X_f$ and $(f_{\operatorname{univ}},X_f)$ intrinsically simulates $(f,X)$.
\end{theorem}

\section{Encodings}
In order to build a universal construction machine for our model, we must be able to encode both the initial graph and the local rule. The next two subsections describe instances of these encodings.

\subsection{Initial graph}
We need an encoding of graphs of ports $\pi$ and labeled over a finite set of labels $\Sigma$. To do so we use the \emph{Depth First Search} (DFS) algorithm for graph exploration to encode our graph into a string over a finite alphabet.

The usual way to encode a graph into a string (or an integer) is to proceed to a DFS on the graph while remembering the edges leading to any previously visited vertex. The alphabet we use to encode the initial graph of a dynamics of degree $|\pi|$ and of labels $\Sigma$ is the following:
$$ \mathcal{A}_{\pi}=\pi^2 \cup \{\$,;,|\} \cup\Sigma$$

With $\Sigma$ the finite set of labels of the vertices and $\{\$,;,|\}$ some arbitrary symbols used as delimiters (we need 3 of them). Figure \ref{fig:ex4} gives an example of the encoding of a graph of order $4$.

\begin{figure}[h]
\begin{center}
\includegraphics[scale=1.5,bb=11 5 89 83]{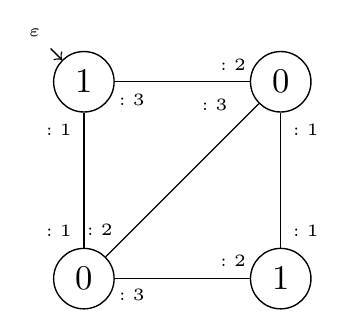}
\end{center}
\caption{Generalized Cayley graph of order $4$. The labels are in the set $\Sigma=\{0,1\}$, numbers on the edges are vertices ports, the incoming arrow on the top-left vertex is the ``starting'' pointer. Its encoding is the string: $\$1;(1,1)\$0;(2,3)\$0(2,3)||;(1,1)\$1(2,3)||;$\label{fig:ex4}}
\end{figure}

The string encoding the graph is a sequence of words, one for each vertex, describing the backward edges ({\emph ie} the edges leading to one of the previously visited vertices) and the forward path leading to the next vertex visited by the DFS.
The structure of a word is described as follow:
$$ \$\sigma\;(i_1,j_1)\overbrace{|\ldots}^{n_1}(i_2,j_2)\overbrace{|\ldots}^{n_2}... ;(s_1,t_1)(s_2,t_2)\ldots (s_n,t_n)\ \$\text{next word}... $$
where:
\begin{itemize}
\item[$\bullet$] $\$$ plays the role of word separator.
\item[$\bullet$] $\sigma\in \Sigma$ is the label of the vertex.
\item[$\bullet$] $(i,j)\overbrace{|\ldots}^{k}$ describes the existence of an edge from port $i$ of the current vertex to port $j$  of the $k^{th}$ vertex when backtracking in the DFS.
\item[$\bullet$] $;$ is a separator between the backward edges and the forward path.
\item[$\bullet$] $(s_1,t_1)(s_2,t_2)\ldots (s_n,t_n) $ describes the path from the current vertex to the next vertex in the DFS.
\end{itemize}
The backwards edge $(i,j)\overbrace{|\ldots}^{k}$ is described using a unary description of the number of times a backtracking has to be made in the DFS. This unary encoding is here to simplify the functioning of the universal machine described in the next section and does not change the time complexity of the construction of the graph.
Notice that as all our graphs are generalized Cayley graphs, they are pointed and thus this encoding is unique (the root of the DFS being the empty path $\varepsilon$). This encoding is not the only way to encode a graph in a string and maybe not the most efficient way but it conveniently fits to our needs while being easy to describe. We could have equivalently defined our encoding using a Breadth First Search algorithm instead of a DFS without changing the complexity of the encoding.

This encoding is both injective and computable (it simply consists in a DFS).

In the following, the encoding of a generalized Cayley graph $X$ in a string is written $\langle X\rangle$.

\subsection{Local Rule}
Here we are interested in the local rules from the set $\mathcal{F}_{\pi,\Sigma,r,b}$. In this section we show that we can give a finite and bounded description of every local rule of this set.
Consider the following properties:
\begin{itemize}
\item $\mathcal{X}^r_{\pi,\Sigma}$ is finite.
\item $\{G\in \mathcal{G}_{\pi,\Sigma}, |V(G)|\leq b\} $ is finite.
\end{itemize}
We have that any $\mathcal{F}_{\pi,\Sigma,r,b}$ is finite. To get a rough bound, let us write $N=|\mathcal{D}^r_{\pi,\Sigma}|$ and $I=|\{G\in \mathcal{G}_{\pi,\Sigma}, |V(G)|\leq b\}|$ then the size of $\mathcal{F}_{\pi,\Sigma,r,b}$ is bounded by $I^N$. Notice that this is just an upper bound and not an exact value. Indeed some disks of $\mathcal{X}^r_{\pi,\Sigma}$ give naming constraints on their images through $f$, limitating the possible images for a given disk $X^r$ to a potentially strict subset of $\{G\in \mathcal{G}_{\pi,\Sigma}, |V(G)|\leq b\}$. \\
Though big, this number is finite and once all parameters $(\pi,\Sigma,r,b)$ are fixed, we can associate to each possible local rule in $\mathcal{F}_{\pi,\Sigma,r,b}$ a description in a finite set of descriptions. We can for instance introduce an ordering on the disks of $\mathcal{D}^r_{\pi,\Sigma}$ and on the graphs of $\{G\in \mathcal{G}_{\pi,\Sigma}, |V(G)|\leq b\}$ and encode the description of $f\in\mathcal{F}_{\pi,\Sigma,r,b}$ in a table of size $N$, such that the $i^{th}$ component of the table contains the index of its image through $f$.

 In the next section we write $\Gamma_{\pi,\Sigma,r,b}$ the finite set of descriptions associated to $\mathcal{F}_{\pi,\Sigma,r,b}$ and $\langle f\rangle$ the description of $f\in\mathcal{F}_{\pi,\Sigma,r,b}$.

\section{Universal dynamics}

\subsection{Universal machine}
The universal machine we design is implemented in the model itself. It consists in a single vertex to which the two encodings are connected.
 The string $\langle X\rangle$ is placed in a linear graph (a sequence of vertices) each vertex containing a symbol of the string. Then this encoding is connected to the machine. Another vertex containing the description of the local rule to simulate is also connected to the machine.

The universal machine can now read the string $\langle X\rangle$ describing $X$ in order to build the graph. While building it, each vertex receives its label $\sigma\in \Sigma$ and the description $\langle f\rangle\in\Gamma_{\pi,\Sigma,r,b}$ of the local rule $f$.
The universal machine itself is a vertex of degree $7$:
\begin{itemize}
\item Two edges for the two inputs ($\langle X\rangle$ and $\langle f\rangle$),
\item Two edges for manipulating the graph being constructed (one pointing at the last added vertex and another traveling along the DFS tree to create backward edges),
\item Two edges for manipulating a stack used to store the sequence of paths linking to consecutive vertices in the DFS (one pointing at the start of the stack, the other reading it),
\item One edge linked to a buffer vertex.
\end{itemize}

The execution of the universal machine is described by a simple local rule acting as the identity everywhere except in the neighborhood of the machine itself.
Notice that the graph generated has labels in $\Sigma \times \Gamma_{\pi,\Sigma,r,b}$. The $\Sigma$ component of the labels contains the label of the vertex and the $\Gamma_{\pi,\Sigma,r,b}$ component contains the description $\langle f\rangle$ of the local rule that has to be simulated. 
More simply, the machine re-constructs the graph $X$ while adding the encoding of $f$ in every vertex label. We denote by $X[\langle f\rangle]$ such a graph.

\subsection{Universal local rule}
We now have to define the rule $f_{\operatorname{univ}}$ that will simulate any rule $f\in\mathcal{F}_{\pi,\Sigma,r,b}$. This rule will be partial and defined only on graphs which are constructed by the universal machine. We can formally define $f_{\operatorname{univ}}$ as follows: for all graph $X\in\mathcal{X}_{\pi,\Sigma}$ and all rule $f\in\mathcal{F}_{\pi,\Sigma,r,b}$:
$$ f_{\operatorname{univ}}(X[\langle f\rangle]_u^r)=f(X_u^r)[\langle f\rangle]  $$
Less formally, $f_{\operatorname{univ}}$ simply looks at the description of $f$ and returns $f(X_u^r)$ while preserving the description $\langle f \rangle$ in every vertex.
The next two properties of $f_{\operatorname{univ}}$ confirm is legitimacy as universal rule.

\begin{lemma}
$f_{\operatorname{univ}}$ is a local rule.
\end{lemma}
\begin{proof}
$f_{\operatorname{univ}}$ inherits from all the properties of all the local rules of the set $\mathcal{F}_{\pi,\Sigma,r,b}$ and thus is a local rule.
\end{proof}
Notice that this result is true only because $f_{\operatorname{univ}}$ is partial. If $f_{\operatorname{univ}}$ is not required to be partial, and is defined over graphs where descriptions of different local rule are present, the subgraphs produced by $f_{\operatorname{univ}}$ on two different disks might be non consistent.

\begin{lemma}
$(f_{\operatorname{univ}},X[\langle f\rangle])$ intrinsically simulates $(f,X)$.
\end{lemma}
\begin{proof}
Consider the function $P_f:\mathcal{G}_{\pi,\Sigma}\rightarrow \mathcal{G}_{\pi,\Sigma\times\Gamma_{\pi,\Sigma,r,b}}$ such that for all graph $X$, $P_f(X)=X[\langle f\rangle]$. This function is computable (its simply consists in adding $\langle f\rangle$ in the label of each vertex), injective, uniformly continuous and we have for all $n\in \mathbb{N}$:
$$ F^n_{\operatorname{univ}}(X[\langle f\rangle])=P_f(F^n(X)) $$
with $F_{\operatorname{univ}}$ the localizable function induced by $f_{\operatorname{univ}}$ and $F$ the localizable function induced by $f$.
\end{proof}
Notice that the simulation occurs without any delay: one iteration in the simulated function corresponds to one iteration in the simulator.

\section{Conclusion and further works}
In this work, we defined two objects. A universal machine that, from a description of a local rule $f$ and a description of an initial generalized Cayley graph $X$, produces a graph $X[\langle f\rangle] $ that contains all the necessary information to simulate the successive iterations of $f$ on $X$. A universal local rule that intrinsically simulates $(f,X)$ when applied on $X[\langle f\rangle]$. We then presented an instance of each of these objects. \medskip

However, these objects can only simulate localizable dynamics induced by a local rule with given parameters $(\pi,\Sigma,r,b)$, and thus this result of intrinsic universality can still be improved. The goal is to design a universal local rule whose parameters do not depend on the parameters of the simulated dynamics. The main difficulty in designing such a local rule is to be able to copy the encoding of the simulated local rule, whatever it is, to the image graph.

\bibliography{biblio}
\bibliographystyle{eptcs.bst}

\end{document}